\documentclass{article}
\usepackage{spconf,amsmath,graphicx}

\usepackage{cite}
\usepackage{amsmath,amssymb,amsfonts}
\usepackage{mathtools}
\usepackage{algorithmic}
\usepackage{graphicx}
\usepackage{textcomp}
\usepackage{xcolor}
\def\BibTeX{{\rm B\kern-.05em{\sc i\kern-.025em b}\kern-.08em
    T\kern-.1667em\lower.7ex\hbox{E}\kern-.125emX}}
\usepackage{bm}
\usepackage{svg}
\usepackage{xcolor}
\usepackage[labelformat=simple]{subcaption}

\usepackage{hyperref}
\usepackage{stfloats}
\usepackage{bm}
\usepackage{amsthm}

\newtheorem{theorem}{Theorem}
\newtheorem{corollary}{Corollary}
\newtheorem{lemma}{Lemma}
\newtheorem{remark}{Remark}

\usepackage[capitalise]{cleveref}
\usepackage[inline]{enumitem}
\usepackage{bbm}

\usepackage{tikz}
\usetikzlibrary{tikzmark,calc,decorations.pathreplacing,patterns}
\usepackage{tikzsymbols}
\usepackage{pgfplots} 
\pgfplotsset{compat=1.16,every x tick label/.append style={font=\tiny},every y tick label/.append style={font=\tiny}}

\makeatletter
\newcounter{savesection}
\newcounter{apdxsection}
\renewcommand\appendix{\par
	\setcounter{savesection}{\value{section}}%
	\setcounter{section}{\value{apdxsection}}%
	\setcounter{subsection}{0}%
	\gdef\thesection{\@Alph\c@section}}
\newcommand\unappendix{\par
	\setcounter{apdxsection}{\value{section}}%
	\setcounter{section}{\value{savesection}}%
	\setcounter{subsection}{0}%
	\gdef\thesection{\@arabic\c@section}}
\makeatother

\usepackage[hang,flushmargin]{footmisc} 
\newcommand\blfootnote[1]{%
  \begingroup
  \renewcommand\thefootnote{}\footnote{#1}%
  \addtocounter{footnote}{-1}%
  \endgroup
}

\newcommand{\txt}[1]{\text{\normalfont #1}}

\DeclareMathOperator{\tx}{t}
\DeclareMathOperator{\rx}{r}
\DeclareMathOperator{\T}{\top}
\DeclareMathOperator{\HT}{\mathrm{H}}
\DeclareMathOperator{\F}{\mathrm{F}}
\DeclareMathOperator*{\argmin}{arg\,min}
\DeclareMathOperator*{\argmax}{arg\,max}
\DeclareMathOperator{\crb}{\txt{CRB}(\omega)}

\DeclareMathOperator{\rank}{\txt{rank}}

\title{Jointly optimal array geometries and waveforms in active sensing:\\
New insights into array design 
via the Cramér-Rao bound 
}
\twoauthors
 {Ids van der Werf, Geert Leus}
	{Delft University of Technology, The Netherlands
 }
 {Robin Rajamäki}
	{Aalto University, Finland
 }
\begin{document}
%
\maketitle
\begin{abstract}
This paper investigates jointly optimal array geometry and waveform designs 
for 
active sensing. 
Specifically, we focus on minimizing the Cramér-Rao lower bound (CRB) of the angle of a single target 
in white Gaussian noise. 
We first find that several array-waveform pairs can yield the same CRB by virtue of
sequences with equal sums of squares, i.e., solutions to certain Diophantine equations. 
Furthermore, we show that under physical aperture and sensor number constraints, the CRB-minimizing receive array geometry is unique, whereas the transmit array can be chosen flexibly. We leverage this freedom to design a novel sparse array geometry that not only minimizes the single-target CRB given an optimal waveform, but also has a nonredundant and contiguous sum co-array---a desirable property when launching independent waveforms, with relevance also to the multi-target case.
\hypersetup{allcolors=black,pdfborder=0 0 0}
\blfootnote{This work was supported by the Netherlands Organisation for Applied Scientific Research, Netherlands Defence Academy (TNO-10026587), as well as projects Business Finland 6G-ISAC, Research Council of Finland FUN-ISAC (359094), and EU Horizon INSTINCT (101139161).}
\end{abstract}
\begin{keywords}
Active sensing, Cramér-Rao lower bound, sparse arrays, waveform design.
\end{keywords}
\vspace{-.2cm}
\section{Introduction}
\label{sec:introduction}

Multisensor active sensing systems have recently experienced an increased research interest due to emerging applications 
such as 
automotive radar, and integrated sensing and communications \cite{sun2020mimoradar,liu2023integrated}. 
An important goal of such systems is 
high spatial resolution, including unambiguous and accurate direction-of-arrival (DoA) estimation, whereas key factors influencing 
performance 
are 
the \emph{array geometry} and \emph{transmit waveforms}. 
Among the numerous optimization criteria considered in literature, 
the Cramér-Rao lower 
bound (CRB) remains a popular choice as it provides 
a fundamental limit on unbiased 
DoA estimation 
performance. 
Past works have focused on 
(sparse array \cite{amin2024sparsearrays}) geometry optimization based on the CRB and related bounds
in both passive \cite{chambers1996temporal,gershman1997anote,roy2013sparsity,pote2019reduced,kokke2023sensor} and active sensing \cite{he2010target,tohidi2019sparse,tabrikian2021cognitive},\cite[Ch.~10]{amin2024sparsearrays}, 
as well as transmit waveform optimization 
\cite{forsythe2005waveform,li2008range,bica2019radar,li2023optimal,vanderwerf2023transmit}, \cite[p.~220]{liu2023integrated}. 
However, jointly optimal array geometries and transmit waveforms that minimize the CRB have not been investigated to the best of our knowledge. 
This paper 
seeks to address this gap by focusing on the single-target CRB. 
The single-target case 
provides valuable insight with relevance also to the 
multi-target case, 
and 
applications such as
beam alignment 
\cite{chiu2019active}, target detection, and tracking 
\cite[pp.~122, 422]{liu2023integrated}.%

The contributions of the paper are as follows. 
Firstly, we show that \emph{multiple} array-waveform pairs can yield \emph{equal} CRBs, which follows from the realization that designing such array configurations corresponds to constructing integer sequences with equal sums of squares---a classical problem in number theory \cite{barnett1942SoS,bradley1998equal}. 
Secondly, we derive the receive array geometry \emph{minimizing} the CRB given a family of optimal waveforms 
corresponding to transmit beamforming, and certain physical constraints on the 
array aperture and number of sensors. 
We find a curious asymmetry between the transmitter and receiver: 
the optimal receive array is unique, whereas the transmit array can be chosen quite freely. 
We then leverage this freedom to design a novel jointly optimal sparse array geometry that also has a contiguous and nonredudant sum co-array. These properties are 
desirable for achieving high target identifability and resolution in the general multi-target case when launching independent waveforms \cite{bekkerman2006target,li2007onparameter}. 




\vspace{-.3cm}
\section{Background}

\subsection{Measurement model}
\label{sec:datamodel}

We consider a monostatic active sensing multiple-input multiple output (MIMO) system consisting of $N_\textrm{t}$ transmit (Tx) sensors collocated with $N_\textrm{r}$ receive (Rx) sensors. Thus, the angle of incidence on the Rx array equals the angle of departure of the Tx array. 
Assuming a single target located in the far field of linear Tx and Rx arrays 
at unknown angle $\omega \in [-\pi,\pi)$,  
a narrowband 
received signal model 
is given by \cite{bekkerman2006target}
%
%
\begin{equation}
\begin{split}
    \mathbf{y} 
    & =(\mathbf{S}\otimes \mathbf{I})(\mathbf{a}_\textrm{t}(\omega)\otimes \mathbf{a}_\textrm{r}(\omega) )\gamma + \mathbf{n},\label{eq:signal_model}
\end{split}
\end{equation}
where $\otimes$ denotes the Kronecker product, $\mathbf{S}\in\mathbb{C}^{T\times N_\textrm{t}}$ is a (known) spatio-temporal Tx waveform matrix, $T\geq 1$ is the waveform length in samples, $\gamma\in\mathbb{C}$ is the unknown reflection coefficient, and  
$\mathbf{n}
\in \mathbb{C}^{N_\textrm{r}T}$ denotes an additive (spatio-temporally white) noise vector whose entries 
follow an i.i.d. circularly symmetric normal distribution with $\mathbb{E}(\mathbf{n}\mathbf{n}^{\HT}) = 
\sigma^2 \mathbf{I}$.
Furthermore, $\mathbf{a}_{\textrm{t}}(\omega) = [e^{j d_\textrm{t}[1] \omega}, \dots, e^{j d_\textrm{t}[N_\textrm{t}] \omega} ]^{\T}$ and $\mathbf{a}_{\textrm{r}}(\omega) = [e^{j d_\textrm{r}[1] \omega}, \dots, e^{j d_\textrm{r}[N_\textrm{r}] \omega} ]^{\T}$ 
represent the steering vectors of the 
Tx and Rx arrays, respectively, whose sensor positions $\mathcal{D}_{\tx}=\{d_{\tx}[n]\}_{n=1}^{N_{\tx}}\subset\mathbb{Z}$ and $\mathcal{D}_{\rx}=\{d_\textrm{r}[m]\}_{m=1}^{N_{\rx}}\subset\mathbb{Z}$ are assumed to lie on a grid of integer multiples of half a carrier wavelength. 
Our goal is to estimate the target angle $\omega$, and understand how the choice of Tx/Rx arrays $\mathcal{D}_{\tx},\mathcal{D}_{\rx}$ and waveform matrix $\mathbf{S}$ impact this task.

\subsection{Single-target CRB and optimal transmit waveform}

The single-target CRB 
of angle $\omega$, assuming the reflection coefficient $\gamma$ and noise power $\sigma^2$ are unknown nuisance parameters, can be show to reduce to \cite{stoica1989music}
\begin{align}
    \txt{CRB}(\omega)
    \!=\!\frac{\sigma^2}{2|\gamma|^2}\|\mathbf{P}_{(\mathbf{S}\otimes \mathbf{I})\mathbf{a}_\textrm{tr}(\omega)}^\perp(\mathbf{S}\otimes \mathbf{I})\mathbf{\dot{a}}_\textrm{tr}(\omega)\|_2^{-2}.\label{eq:CRB}
\end{align}
Here, $\mathbf{P}_{\mathbf{X}}^\perp$ denotes the projection onto the orthogonal complement of the range space of $\mathbf{X}$; $\mathbf{a}_\textrm{tr}(\omega) \triangleq \mathbf{a}_\textrm{t}(\omega) \otimes \mathbf{a}_\textrm{r}(\omega)$ is the effective Tx-Rx steering vector; and $\mathbf{\dot{a}}_\textrm{tr}(\omega) \triangleq \frac{\partial}{\partial\omega}\mathbf{a}_\textrm{tr}(\omega)$ is its derivative with respect to $\omega$. 
%
Fortsythe and Bliss \cite{forsythe2005waveform}, as well as Li \emph{et al.} \cite{li2008range}, 
investigated 
waveforms $\mathbf{S}$ minimizing the 
CRB 
\emph{given} an array geometry. 
In the single-target case \eqref{eq:CRB}, 
the optimal waveform depends on the \emph{``spatial variances"} of the Tx and Rx arrays \cite{forsythe2005waveform}, $\chi_{\tx}\!\triangleq\!\chi(\mathcal{D}_{\tx})$ and $\chi_{\rx}\!\triangleq\!\chi(\mathcal{D}_{\rx})$, where
\begin{align}
    \chi(\mathcal{D})\triangleq\frac{1}{|\mathcal{D}|}
            \sum_{d\in\mathcal{D}}(d-\mu(\mathcal{D}))^2,
            \label{eq:sos_t_r}
\end{align}
and the corresponding spatial mean 
is $\mu(\mathcal{D})\triangleq \frac{1}{|\mathcal{D}|}\sum_{d\in\mathcal{D}}d$. In particular, if the following condition is satisfied:
\begin{equation}
 \chi_{\rx} > \chi_{\tx},
\label{eq:sb_cond}
\end{equation}
then the optimal waveform matrix has the form 
\cite{forsythe2005waveform}:
\begin{align}
    \mathbf{S}_o
        \!\triangleq\!
\underset{\mathbf{S}\in\mathbb{C}^{T\times N_{\tx}}}{\argmin} \{\textrm{CRB}(\omega)\!:\!  \chi_{\rx}\!>\!\chi_{\tx},\!\|\mathbf{S}\|_{\txt{F}}^2\!\leq\!1\}
\!=\!
\frac{\mathbf{u}\mathbf{a}_{\tx}^{\HT}(\omega)}{\sqrt{N_{\tx}}},
    \label{eq:optTxWaveform_sum}
\end{align}
where $\mathbf{u}\in\mathbb{C}^T$ is an arbitrary unit norm vector ($\|\mathbf{u}\|_2=1$), and the Tx power $\|\mathbf{S}\|_\textrm{F}^2$ is w.l.o.g. constrained to $\leq 1$. \cref{eq:optTxWaveform_sum} simply corresponds to fully coherent transmission in the target direction $\omega$, i.e., \emph{Tx beamforming}. If $\chi_{\rx}=\chi_{\tx}$, then optimal waveforms beyond \eqref{eq:optTxWaveform_sum} also exist \cite{li2008range}. Otherwise, if $\chi_{\rx}<\chi_{\tx}$, then the optimal waveform corresponds to transmitting infinitesimal energy in the target direction---see \cite{forsythe2005waveform,li2008range} for details. As this solution has limited practical relevance, we will henceforth focus on \eqref{eq:optTxWaveform_sum} which is optimal given \eqref{eq:sb_cond}. 

We conclude by highlighting two interesting facts revealed by \eqref{eq:optTxWaveform_sum}. 
Firstly, 
any optimal waveform matrix 
is column rank-deficient (when $N_{\tx}\!>\!1$). 
Hence, widely employed \emph{orthogonal waveforms do not generally minimize the CRB}---even in the case of multiple targets \cite{li2008range}. 
Secondly, any optimal waveform depends on the Tx array geometry $\mathcal{D}_{\tx}$ (and true target angle $\omega$) via steering vector $\mathbf{a}_{\tx}(\omega)$. 
That is, \emph{different Tx arrays lead to different optimal waveforms}. While this was observed in \cite{forsythe2005waveform,li2008range},
the impact of the array geometry on the CRB was not fully explored. Hence, we attempt to fill this gap by asking 
    \emph{which Tx/Rx array geometries minimize the 
    CRB in \eqref{eq:CRB} jointly with the optimal waveform in \eqref{eq:optTxWaveform_sum}?}

\section{Jointly optimal array-waveform pairs} \label{sec:main}

\subsection{Equal CRB via Rx arrays with equal sums of squares}
Substituting the optimal waveform in \eqref{eq:optTxWaveform_sum} into \eqref{eq:CRB} can be shown to simplify the single-target CRB into \cite{forsythe2005waveform}
\begin{align}
    \text{CRB}(\omega) = 
    \frac{\sigma^2}{2|\gamma|^2}\frac{1}{N_{\tx}N_{\rx}}\chi_{\rx}^{-1}.
    \label{eq:CRB_sb}
\end{align}
Hence, the CRB (given an optimal waveform) is independent of the target angle $\omega$ and only depends on the Tx array geometry via the number of Tx sensors, $N_{\tx}$. In contrast, for a fixed number of Rx sensor $N_{\rx}$, the CRB depends on the Rx array geometry via its spatial variance, $\chi_{\rx}$. 
This suggests an intriguing possibility: Rx array geometries with equal spatial variances yield equal CRBs. Designing such arrays actually corresponds to finding integer sequences with \emph{equal sums of squares}---a special class of Diophantine equations that have a long history in number theory 
\cite{barnett1942SoS,bradley1998equal}. For example, $1^2\!+\!8^2\!=\!4^2\!+\!7^2$ can be used to construct arrays $\mathcal{D}_1=\{-8,-1, 1,8\}$ and $\mathcal{D}_2=\{-7,-4,4,7\}$ 
satisfying $\chi(\mathcal{D}_1)=\chi(\mathcal{D}_2)$. While \emph{different} array geometries can achieve the \emph{same} CRB, their practical DoA estimation performance may differ, as \cref{sec:num} will demonstrate. 
Fully exploring this prospect is left for future work. Instead, we now turn our attention to deriving the Rx array configuration minimizing the CRB \eqref{eq:CRB_sb}. 

\subsection{Optimal Rx array geometry: Clustered array}\label{subsec:restrictedAperture}
 
The CRB in \eqref{eq:CRB_sb} is a monotonically decreasing function in increasing  $\chi_{\rx}$. 
Since the CRB can be made arbitrarily small simply by expanding the Rx aperture without bound, a more meaningful question is: which array geometry is optimal under a constraint on the physical array aperture?
To answer this question, we make use of the following \namecref{thm:optimal_subset}, which shows that the sparse array geometry whose sensors are clustered around the extremes of the array maximizes the spatial variance under an aperture constraint. 
For simplicity and brevity, we denote the set of nonnegative integers smaller than $N$, 
i.e., the $N$-sensor uniform linear array (ULA), by $\mathcal{U}_N\!\triangleq\!\{0,1,\ldots,N\!-\!1\}$,
and focus on the case of \emph{even} $N$.
\begin{lemma}[Clustered array] \label{thm:optimal_subset}
Let $\mathcal{L} = \mathcal{U}_{L+1}$, where $L\in\mathbb{N}_+$ is fixed. Then, given an even $N\leq L+1$, 
the subset $\mathcal{D} \subseteq \mathcal{L}$ of size $|\mathcal{D}|=N$ maximizing $\chi(\mathcal{D})$ in \eqref{eq:sos_t_r} 
is $\mathcal{D}=\mathcal{K}_N^L$, where
\begin{align}
    \mathcal{K}_N^L
    \!\triangleq\!\underset{\mathcal{D}\subseteq \mathcal{L}}{\argmin}\{\chi(\mathcal{D})\!:\!|\mathcal{D}|\!=\!N\}
    =\mathcal{U}_{N/2}\!\cup\!(L-\mathcal{U}_{N/2}).
    \label{eq:clustered}
\end{align}
Moreover, the (optimal) value of the spatial variance is
\begin{align}
    \chi(\mathcal{K}_N^L)
    =\tfrac{1}{4}((L+1-\tfrac{N}{2})^2+\tfrac{1}{3}(\tfrac{N^2}{4}-1)).
    \label{eq:clustered_obj}
\end{align}
\end{lemma}
The proof of \eqref{eq:clustered} follows directly 
via negation 
and is 
omitted for 
brevity. 
Similarly, the value of $\chi(\mathcal{K}_N^L)$ follows by straightforward computation after 
substituting \eqref{eq:clustered} into \eqref{eq:sos_t_r}.

We can now characterize the set of array-waveform pairs jointly minimizing the single-target CRB \eqref{eq:CRB}, i.e., solutions to
    \begin{align}
\underset{\substack{\mathcal{D}_{\tx},\mathcal{D}_{\rx}\subset \mathbb{N}\\ \mathbf{S}\in\mathbb{C}^{T\times N_{\tx}}}}{\txt{minimize}}\ \crb\
\txt{s.t.}\ 
\begin{cases}
        |\mathcal{D}_{\tx}|=N_{\tx}, |\mathcal{D}_{\rx}|=N_{\rx},\\
        \|\mathbf{S}\|_{\F}^2\leq 1, \\ 
        \max \mathcal{D}_{\rx}\!-\!\min \mathcal{D}_{\rx}\leq L,\\
        \chi(\mathcal{D}_{\rx})>\chi(\mathcal{D}_{\tx}).
        \end{cases}\label{eq:joint_opt}
    \end{align}
In addition to the number of Tx/Rx sensors and Tx power, we have also constrained the Rx aperture (to $\leq L$). Moreover, the spatial variance of the Tx array should be 
smaller
than that of the Rx array to ensure that Tx beamforming \eqref{eq:optTxWaveform_sum} is optimal.

\begin{theorem}[Jointly optimal array-waveform pairs]\label{thm:optimal_aperture}
    Suppose \eqref{eq:joint_opt} is feasible for given $N_{\tx},N_{\rx}, L\in\mathbb{N}_+$, where 
    $N_{\rx}$ is even. 
    Then the solutions $(\mathcal{D}_{\tx}^\star,\mathcal{D}_{\rx}^\star,\mathbf{S}^\star)$ 
    to \eqref{eq:joint_opt} are: 
    $\mathbf{S}^\star$ 
    given by \eqref{eq:optTxWaveform_sum}, 
    \begin{align}
        \mathcal{D}_{\rx}^\star=\mathcal{K}_{N_{\rx}}^L\label{eq:opt_rx_array}
    \end{align}
    given by \eqref{eq:clustered}, 
    and any $\mathcal{D}_{\tx}^\star$ satisfying \eqref{eq:sb_cond} and $|\mathcal{D}_{\tx}^\star|=N_{\tx}$.
\end{theorem}
\begin{proof}
Per \cite{forsythe2005waveform}, 
\eqref{eq:optTxWaveform_sum} is optimal when \eqref{eq:sb_cond} holds. Hence, the CRB \eqref{eq:CRB} simplifies to \eqref{eq:CRB_sb}. By \cref{thm:optimal_subset}, \eqref{eq:CRB_sb} is minimized by 
\eqref{eq:opt_rx_array}, 
when the Rx aperture is $\leq L$. Finally, any 
$\mathcal{D}_{\tx}$ satisfying \eqref{eq:sb_cond} and $|\mathcal{D}_{\tx}|\!=\!N_{\tx}$ has 
the same CRB, and is thus optimal.
\end{proof}
\begin{remark}\label{rem:uniqueness}
    By \cref{thm:optimal_aperture}, the clustered array in \eqref{eq:opt_rx_array} is the \underline{unique} optimal Rx array, whereas several optimal waveforms and Tx arrays exist: 
    $\mathbf{S}^\star$ 
    follows \eqref{eq:optTxWaveform_sum} and therefore depends on $\mathcal{D}_{\tx}^\star$, which can be chosen freely provided it satisfies \eqref{eq:sb_cond}.
\end{remark}
The subtle question remains for which values of tuple $(N_{\tx},N_{\rx},L)$ optimization problem \eqref{eq:joint_opt} is feasible? A simple sufficient condition is $N_{\tx}\leq N_{\rx}\leq L+1$, since then the feasible set of \eqref{eq:joint_opt} contains $\mathcal{D}_{\rx}=\mathcal{U}_{N_{\rx}}$ and $\mathcal{D}_{\tx}=\mathcal{U}_{N_{\tx}}$, which trivially satisfy \eqref{eq:sb_cond}. 
Deriving necessary and sufficient conditions 
in terms of $(N_{\tx},N_{\rx},L)$ 
is 
beyond the scope of this paper and left for future work. 
Instead, we now turn our attention to how to pick \emph{an} optimal Tx array among the possible choices. 

\subsection{How should the Tx array geometry be chosen?}\label{sec:tx}

The minimum single-target CRB---jointly optimized over the waveform and array geometry---may be achieved by multiple choices of the Tx array (cf. \cref{rem:uniqueness}). Nevertheless, some Tx array configurations might be preferable over others in terms of performance indicators beyond the CRB. Herein, we consider \emph{identifiability}, which quantifies if for a fixed number of $K$ targets, any given (noiseless) measurement can be associated with a unique set of DoAs. 
In active sensing, identifiability depends on the geometry of the sum co-array \cite{li2007onparameter}:
\begin{align*}
    \mathcal{D}_\Sigma\triangleq\mathcal{D}_{\tx}+\mathcal{D}_{\rx}=\{d_{\tx}+d_{\rx}\ |\ d_{\tx}\in\mathcal{D}_{\tx};d_{\rx}\in\mathcal{D}_{\rx}\}.
\end{align*}
Let $N_\Sigma\triangleq|\mathcal{D}_\Sigma|$ denote the number of sum co-array elements. A sufficient condition for identifying any $K\leq N_\Sigma/2$ targets
is that the sum co-array is \emph{contiguous} and the Tx waveform has full column rank, i.e., $\mathcal{D}_\Sigma\!=\!\mathcal{U}_{N_{\Sigma}}$ 
and $\rank(\bm{S})\!=\!N_{\tx}$ \cite{rajamaki2023importance}. 
Since $N_\Sigma\leq N_{\tx}N_{\rx}$, up to $ N_{\tx}N_{\rx}/2$ targets can be identified by a contiguous sum co-array with appropriately chosen (e.g., orthogonal) waveforms.\footnote{Fully leveraging the sum co-array would hence in general require transmitting a suboptimal waveform, i.e., one that does not minimize the CRB.} 
An array achieving $N_\Sigma\!=\!N_{\tx}N_{\rx}$ is called \emph{nonredundant}. 
The following \namecref{THM:CONTIGUOUS} establishes a set of $(N_{\tx},N_{\rx},L)$ tuples solving \eqref{eq:joint_opt} and yielding a contiguous and nonredundant sum co-array. For a proof, see Appendix~\labelcref{a:proof}.
\begin{corollary}[Optimal contiguous nonredundant co-array]\label{THM:CONTIGUOUS}
    Given $N_{\tx},N_{\rx}\in\mathbb{N}_+$, where $N_{\rx}$ is even, if $L=(N_{\tx}+1)N_{\rx}/2-1$, then the following Tx array is a solution to \eqref{eq:joint_opt}:  
    \begin{align}
        \mathcal{D}_{\tx}^\star=
        \tfrac{N_{\rx}}{2}\mathcal{U}_{N_{\tx}}.
        \label{eq:opt_tx_array}
    \end{align}
    The sum co-array of \eqref{eq:opt_tx_array} and the optimal Rx array 
    \eqref{eq:opt_rx_array} 
    is contiguous ($\mathcal{D}_\Sigma=\mathcal{U}_{N_\Sigma}$) and nonredundant ($N_\Sigma=N_{\tx}N_{\rx}$).
\end{corollary}

To the best of our knowledge, the sparse Tx-Rx array geometry in 
\labelcref{eq:opt_rx_array,eq:opt_tx_array} has not appeared in the literature before. Neither has its optimality w.r.t. minimizing the single-target CRB been established, nor the fact that it
can achieve 
a contiguous nonredundant sum co-array. 
We note that the clustered array in \eqref{eq:clustered} has empirically \cite{gershman1997anote} been found to minimize the so-called \emph{unconditional} 
CRB corresponding to a different (single-target) measurement model 
typically arising in passive sensing. This is nevertheless different from the active sensing model \eqref{eq:signal_model} and \emph{conditional} 
CRB considered herein, and hence does not imply our results. 
Moreover, in stark contrast to passive sensing, in active sensing one has the freedom to choose both the Tx array geometry (as in \cref{sec:tx}) and the transmitted waveforms, despite the optimal Rx array being a clustered array in both cases. 
Exploring generalizations of the array geometry in 
\labelcref{eq:opt_rx_array,eq:opt_tx_array} 
is left for future work. Such generalizations, possibly with a redundant or even noncontiguous sum co-array, may be of interest in minimizing the CRB given an 
$\mathbf{S}$ 
differing from the (optimal) choice in \eqref{eq:optTxWaveform_sum}.

\vspace{-0.3cm}
\section{Numerical examples}\label{sec:num}
\noindent Next, we illustrate the results of \cref{sec:main} numerically. We focus on the four array geometries depicted in \cref{fig:arrays}, where \subref{fig:clust} shows the optimal array defined by \labelcref{eq:opt_rx_array,eq:opt_tx_array}.
We consider the maximum-likelihood estimator (MLE) of $\omega$ for a \emph{fixed} waveform matrix $\mathbf{S}$. Given \eqref{eq:signal_model}, the MLE can be shown to reduce to the following joint Tx-Rx beamformer \cite{evans1982application}: 
$\hat{\omega}=\argmax_{\bar{\omega}\in[-\pi,\pi)}|\mathbf{y}^{\HT}(\mathbf{S}\mathbf{a}_{\tx}(\bar{\omega})\otimes \mathbf{a}_{\rx}(\bar{\omega}))|^2/\|\mathbf{S}\mathbf{a}_{\tx}(\bar{\omega})\|_2^2$. 
We assume the optimal waveform in \eqref{eq:optTxWaveform_sum} is used, 
with
$\mathbf{u}=\tfrac{1}{\sqrt{T}}\mathbf{1}_T$
and $T=N_{\tx}$. 
Since $\mathbf{S}$ is a function of the true angle, an initial estimate of $\omega$ would be needed in practice; 
see \cite{li2008range} for a discussion and examples. 
The ground truth target angle and reflectivity are set to $\omega=0$ and $\gamma=1$, respectively, whereas entries of noise vector $\mathbf{n}$ are drawn from an i.i.d. circularly symmetric normal distribution with variance $\sigma^2$. The squared error of the MLE is averaged over $10^4$ Monte Carlo trials.
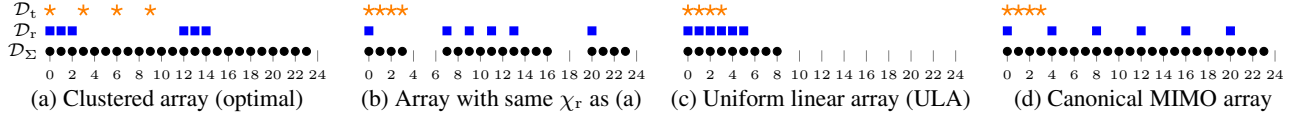
\begin{figure*}
    \newcommand{\dataLoc}{Data/ClusteredSpecialCase}
    \newcommand{\xmax}{24}
    \newcommand{\xmin}{0}
    \newcommand{\fwidth}{5.2}
    \newcommand{\msize}{1.5}
	\pgfplotsset{every x tick label/.append style={font=\tiny},every y tick label/.append style={font=\scriptsize}}
	\centering
	\subfloat[Clustered array (optimal)]{   
		\begin{tikzpicture}
			\begin{axis}[width= \fwidth cm ,height=2.4 cm,xmin=\xmin-0.2,xmax=\xmax+0.2,xtick={\xmin,\xmin+2,...,\xmax},ytick={-1,0,1},yticklabels={$\mathcal{D}_\Sigma$,$\mathcal{D}_{\rx}$,$\mathcal{D}_{\tx}$},ytick style={draw=none},ymin=-1.5,ymax=1.5,xticklabel shift = 0 pt,xtick pos=bottom,axis line style={draw=none}]
				\addplot[orange,only marks,line width = 0.7,mark=star,mark size=2.2,y filter/.code={\pgfmathparse{\pgfmathresult*0+1}}] table[x=D,y=D]{\dataLoc/Dt.dat};
				\addplot[blue,only marks,mark=square*,mark size=\msize,y filter/.code={\pgfmathparse{\pgfmathresult*0}}] table[x=D,y=D]{\dataLoc/Dr.dat};
				\addplot[only marks,mark=*,mark size=\msize,y filter/.code={\pgfmathparse{\pgfmathresult*0-1}},x filter/.code={\pgfmathparse{\pgfmathresult}}] table[x=D,y=D]{\dataLoc/D_Sigma.dat};
			\end{axis}
		\end{tikzpicture}\label{fig:clust}
  \vspace{-0.25cm}
	}%
 \renewcommand{\dataLoc}{Data/ArrayF}
	\centering
	\subfloat[Array with same $\chi_{\rx}$ as \subref{fig:clust}]{    
		\begin{tikzpicture}
			\begin{axis}[width= \fwidth cm ,height=2.4 cm,xmin=\xmin-0.2,xmax=\xmax+0.2,xtick={\xmin,\xmin+2,...,\xmax},ytick={-1,0,1},
   yticklabels={},
   ytick style={draw=none},ymin=-1.5,ymax=1.5,xticklabel shift = 0 pt,xtick pos=bottom,axis line style={draw=none}]
				\addplot[orange,only marks,line width = 0.7,mark=star,mark size=2.2,y filter/.code={\pgfmathparse{\pgfmathresult*0+1}}] table[x=D,y=D]{\dataLoc/Dt.dat};
				\addplot[blue,only marks,mark=square*,mark size=\msize,y filter/.code={\pgfmathparse{\pgfmathresult*0}}] table[x=D,y=D]{\dataLoc/Dr.dat};
				\addplot[only marks,mark=*,mark size=\msize,y filter/.code={\pgfmathparse{\pgfmathresult*0-1}},x filter/.code={\pgfmathparse{\pgfmathresult}}] table[x=D,y=D]{\dataLoc/D_Sigma.dat};
			\end{axis}
		\end{tikzpicture}\label{fig:alt}
  \vspace{-0.25cm}
	}%
    \renewcommand{\dataLoc}{Data/ULA}
	\centering
	\subfloat[Uniform linear array (ULA)]{    
		\begin{tikzpicture}
			\begin{axis}[width= \fwidth cm ,height=2.4 cm,xmin=\xmin-0.2,xmax=\xmax+0.2,xtick={\xmin,\xmin+2,...,\xmax},ytick={-1,0,1},
   yticklabels={},
   ytick style={draw=none},ymin=-1.5,ymax=1.5,xticklabel shift = 0 pt,xtick pos=bottom,axis line style={draw=none}]
				\addplot[orange,only marks,line width = 0.7,mark=star,mark size=2.2,y filter/.code={\pgfmathparse{\pgfmathresult*0+1}}] table[x=D,y=D]{\dataLoc/Dt.dat};
				\addplot[blue,only marks,mark=square*,mark size=\msize,y filter/.code={\pgfmathparse{\pgfmathresult*0}}] table[x=D,y=D]{\dataLoc/Dr.dat};
				\addplot[only marks,mark=*,mark size=\msize,y filter/.code={\pgfmathparse{\pgfmathresult*0-1}},x filter/.code={\pgfmathparse{\pgfmathresult}}] table[x=D,y=D]{\dataLoc/D_Sigma.dat};
			\end{axis}
		\end{tikzpicture}\label{fig:ula}
  \vspace{-0.25cm}
	}%
 \renewcommand{\dataLoc}{Data/Nested}
	\centering
	\subfloat[Canonical MIMO array]{
		\begin{tikzpicture}
			\begin{axis}[width= \fwidth cm ,height=2.4 cm,xmin=\xmin-0.2,xmax=\xmax+0.2,xtick={\xmin,\xmin+2,...,\xmax},ytick={-1,0,1},
   yticklabels={},
   ytick style={draw=none},ymin=-1.5,ymax=1.5,xticklabel shift = 0 pt,xtick pos=bottom,axis line style={draw=none}]
				\addplot[orange,only marks,line width = 0.7,mark=star,mark size=2.2,y filter/.code={\pgfmathparse{\pgfmathresult*0+1}}] table[x=D,y=D]{\dataLoc/Dt.dat};
				\addplot[blue,only marks,mark=square*,mark size=\msize,y filter/.code={\pgfmathparse{\pgfmathresult*0}}] table[x=D,y=D]{\dataLoc/Dr.dat};
				\addplot[only marks,mark=*,mark size=\msize,y filter/.code={\pgfmathparse{\pgfmathresult*0-1}},x filter/.code={\pgfmathparse{\pgfmathresult}}] table[x=D,y=D]{\dataLoc/D_Sigma.dat};
			\end{axis}
		\end{tikzpicture}\label{fig:nest}
  \vspace{-0.25cm}
	}%
 \vspace{-0.25cm}
\caption{Array geometries with $N_{\tx}\!=\!4$ transmitter and $N_{\rx}\!=\!6$ receiver sensors. 
Array \subref{fig:clust} 
maximizes the Rx spatial variance $\chi_{\rx}$
given 
physical aperture $L\!=\!14$, and has a contiguous nonredudant sum co-array. 
The Rx array in \subref{fig:alt} has equal $\chi_{\rx}$, but larger $L$.
}\label{fig:arrays}
\vspace{-.5cm}
\end{figure*}

\cref{fig:mse} shows the (single-target) CRBs and empirical MLE performance of the array configurations in \cref{fig:arrays} as a function of $\txt{SNR}\triangleq 10\log(|\gamma|^2/\sigma^2$). 
The optimal array in \cref{fig:clust} minimizes the CRB among all geometries with aperture $L\!=\!14$. 
By using a larger aperture, one can construct array configurations achieving \emph{equal} or larger CRB, as the array in \cref{fig:alt} with $L=20$ shows. Indeed, the spatial variances of the Rx arrays in \subref{fig:clust} and \subref{fig:alt} are identical by virtue of the following 
equal sums of squares: 
$5^2+6^2+7^2=1^2+3^2+10^2$. 
Although the CRBs of arrays \subref{fig:clust} and \subref{fig:alt} are identical, their MLE performance in the threshold region differs. This is related to the difference in the beampatterns of the two arrays.
Understanding exactly how the array geometry affects MLE, especially for arrays with equal spatial variance in \eqref{eq:sos_t_r}, is an open question. 
Finally, we contrast the clustered array in \cref{fig:arrays}~\subref{fig:clust} to the well-known  \subref{fig:ula} ULA and \subref{fig:nest} canonical MIMO (radar) array with a nested structure. 
The ULA has a significantly higher CRB than the other arrays due to its smaller spatial variance. In contrast, the MIMO 
array, due to its larger Rx aperture ($L=20$), has a slightly lower CRB than the clustered array in \subref{fig:clust}. 
However, the clustered array requires \emph{only half the physical aperture} to realize the same 
sum co-array. 
This is advantageous when high identifiability and angular resolution 
yet small physical array size is desirable, 
as in automotive radar \cite{sun2020mimoradar}. 
%
The MLE of the MIMO array suffers from poor performance due to spatial aliasing,  
as its Rx array (a dilated ULA) lacks consecutive elements, unlike the Clustered array. 
Although the issue can be alleviated by restricting the search space of the MLE to the vicinity of $\omega$ based on the region illuminated upon Tx, 
it also illustrates that upon coherent transmission, 
angle estimation performance heavily depends on the Rx array geometry, as opposed to the 
sum co-array which is key when transmitting independent 
waveforms \cite{li2007onparameter}. Judicious array design is thus needed to ensure robust performance across various transmission strategies.
\vspace{-0.3cm}
\section{Conclusion}
\label{sec:conclusion}

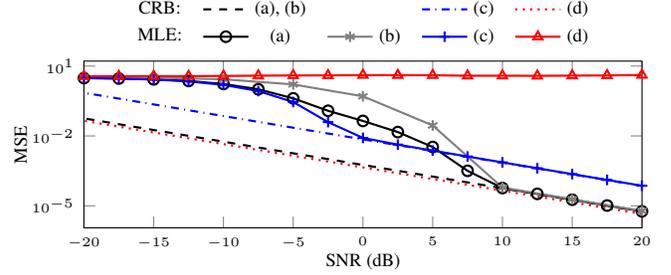
\begin{figure}
    	\centering
    	\begin{tikzpicture}
    		\begin{axis}[width=9 cm,height= 3.8 cm,ylabel={\scriptsize MSE},xlabel= {\scriptsize SNR (dB)},
      xmin=-20,xmax=20,
      ymode=log,
      xlabel shift = {-5 pt},
      ylabel shift = {-5 pt},
      legend style = {at={(0.5,1.03)},anchor=south,draw=none,fill=none},legend columns=5,legend style={/tikz/every even column/.append style={column sep=0.2cm,font={\scriptsize}}}]
            \addlegendimage{empty legend}
         \addlegendentry{CRB:}
           \addplot[black,dashed,thick,draw] table[x=SNR,y=CRB]{Data/MSE_CRB_clustered_optimal.dat};
           \addlegendentry{\subref{fig:clust}, \subref{fig:alt}}
           \addlegendimage{empty legend}
         \addlegendentry{}
           \addplot[blue,dashdotted,thick,draw] table[x=SNR,y=CRB]{Data/MSE_CRB_ULA_optimal.dat};
           \addlegendentry{\subref{fig:ula}}
            \addplot[red,dotted,thick,draw] table[x=SNR,y=CRB]{Data/MSE_CRB_nested_optimal.dat};
           \addlegendentry{\subref{fig:nest}}
            \addlegendimage{empty legend}
         \addlegendentry{MLE:}
          \addplot[black,thick,draw,mark=o] table[x=SNR,y=MSE]{Data/MSE_CRB_clustered_optimal.dat};
           \addlegendentry{\subref{fig:clust}}
         \addplot[gray,thick,draw,mark=asterisk,mark options={solid}] table[x=SNR,y=MSE]{Data/MSE_CRB_array_G_optimal.dat};
           \addlegendentry{\subref{fig:alt}}
            \addplot[blue,thick,draw,mark=+,mark options={solid}] table[x=SNR,y=MSE]{Data/MSE_CRB_ULA_optimal.dat};
           \addlegendentry{\subref{fig:ula}}
           \addplot[red,thick,draw,mark=triangle,mark options={solid}] table[x=SNR,y=MSE]{Data/MSE_CRB_nested_optimal.dat};
           \addlegendentry{\subref{fig:nest}}
    		\end{axis}%
    	\end{tikzpicture}\vspace{-0.3cm}
     \caption{Single-target CRB and MLE performance of array geometries in \cref{fig:arrays} using an optimal 
     waveform following \eqref{eq:optTxWaveform_sum}. 
     }\label{fig:mse}
     \vspace{-.3cm}
\end{figure}

This paper investigated waveform-array geometry pairs minimizing the single-target CRB in active sensing. 
Focusing on a family of optimal waveforms \cite{forsythe2005waveform,li2008range} corresponding to Tx beamforming in the target direction, 
we showed that the optimal linear Rx array places sensors at the edges of its aperture 
to maximize its spatial variance---the sums of squares of its centered sensor positions. The Tx array geometry can be chosen freely, provided its spatial variance does not exceed that of the Rx array. We established that 
the Tx array can be selected such that the sum co-array of the joint Tx-Rx array 
is both contiguous and nonredundant. 
The derived array geometry therefore has optimal properties both (in the single and multi-target cases) when launching coherent and independent waveforms, in contrast to the ULA, which has substantially higher CRB and lower identifiability, and the (nonredundant nested) MIMO array, which requires a larger physical aperture to achieve a comparable CRB and sum co-array. 


\vspace{-.3cm}
\appendix
\section{Proof of Corollary~\ref{THM:CONTIGUOUS}}\label{a:proof}
    We first show that if $L=(N_{\tx}+1)N_{\rx}/2-1$ 
    then \eqref{eq:opt_tx_array} is a solution of \eqref{eq:joint_opt}. This reduces to showing that \eqref{eq:opt_tx_array} satisfies \eqref{eq:sb_cond} given $\mathcal{D}_{\rx}^\star=\mathcal{K}_{N_{\rx}}^L$, which \cref{thm:optimal_aperture} established was the optimal Rx array configuration. By \eqref{eq:sos_t_r}, we have
    \begin{align}
        \chi(\mathcal{D}_{\tx}^\star)
        \!=\!
        \chi(\tfrac{N_{\rx}}{2}\mathcal{U}_{N_{\tx}})
        \!=\!
        \tfrac{N_{\rx}^2}{4}\chi(\mathcal{U}_{N_{\tx}})\!=\!
        \tfrac{1}{48}(N_{\tx}^2-1)N_{\rx}^2.\label{eq:sos_opt_tx}
    \end{align}
    Since $\mathcal{D}_{\rx}^\star=\mathcal{K}_{N_{\rx}}^L$, where $L=(N_{\tx}+1)N_{\rx}/2-1$, condition $\chi(\mathcal{D}_{\rx}^\star)> \chi(\mathcal{D}_{\tx}^\star)$ can be rewritten using \labelcref{eq:sos_opt_tx,eq:clustered_obj} as
    \begin{align*}
        \tfrac{1}{4}((\tfrac{N_{\tx}N_{\rx}}{2})^2+\tfrac{1}{3}(\tfrac{N_{\rx}^2}{4}-1))\!>\! 
        \tfrac{1}{48}(N_{\tx}^2-1)N_{\rx}^2.
    \end{align*}
    Rearranging terms yields $2N_{\rx}^2(N_{\tx}^2+1)-4> 0$, 
    which holds for any even $N_{\rx}\geq 2$. Hence, \eqref{eq:sb_cond} is satisfied, which implies that the feasible set of \eqref{eq:joint_opt} is nonempty, and per \cref{thm:optimal_aperture}, that \eqref{eq:opt_tx_array} is an optimal solution to \eqref{eq:joint_opt}.

    We now show that the sum co-array 
    is contiguous and nonredundant. Let $\alpha\!=\!N_{\rx}/2$ and $\beta\!=\!N_{\tx}\alpha\!=\!N_{\tx}N_{\rx}/2$. Then
    \begin{align*}
    \mathcal{D}_\Sigma&=\alpha \mathcal{U}_{N_{\tx}}+(\mathcal{U}_{\alpha}\cup(L-\mathcal{U}_{\alpha}))
    \\
    &=(\alpha \mathcal{U}_{N_{\tx}}+\mathcal{U}_{\alpha})\cup(\alpha \mathcal{U}_{N_{\tx}}-\mathcal{U}_{\alpha}+L),
    \end{align*}
    where $\alpha \mathcal{U}_{N_{\tx}}+\mathcal{U}_{\alpha} = \{\alpha m+n\ |\ m\in \mathcal{U}_{N_{\tx}}; n\in\mathcal{U}_{\alpha}\}=\mathcal{U}_{\beta}$. 
    Furthermore, note that $-\mathcal{U}_M=\mathcal{U}_M-M+1$. Hence, 
\begin{align*}
    \alpha \mathcal{U}_{N_{\tx}}- \mathcal{U}_{\alpha}+L=\alpha \mathcal{U}_{N_{\tx}}+ \mathcal{U}_{\alpha}+L-\alpha+1.
\end{align*}
 Recalling that $L=\beta+\alpha-1$ then yields the desired result
    \begin{align*}
        \mathcal{D}_\Sigma
        =\mathcal{U}_{\beta}\cup(\mathcal{U}_{\beta}+\beta)=\mathcal{U}_{2\beta}=\mathcal{U}_{N_{\tx}N_{\rx}}.\qquad \square
    \end{align*}

\vfill\pagebreak
\clearpage

\unappendix	

\bibliographystyle{IEEEbib}

\bibliography{referencesAbbrvd} 


\end{document}